\documentclass[12pt]{article}
\usepackage{graphicx}
\usepackage{amssymb}
\usepackage{amsmath}
\newtheorem{thm}{Theorem}[section]
\newtheorem{cor}[thm]{Corollary}

\numberwithin{equation}{section}

\newcommand{\eh}{\hfill}\newlength{\sperr}

\newenvironment{proof}{{\settowidth{\sperr}{\bf\rm
Proof}%
\par\addvspace{0.3cm}\noindent\parbox[t]{1.3\sperr}
{\textit{ P\eh r\eh o\eh o\eh f\eh .}}%
}}{\nopagebreak\mbox{}
$\Box$\par\addvspace{0.3cm}}

\newtheorem{Pa}{Paper}[section]

\newtheorem{Pn}[Pa]{{\bf Proposition}}
\title{Algorithmic entropy, thermodynamics, \\ and game interpretation}
\author{Lev Sakhnovich}
\date{}
\begin{document}
\maketitle

\thanks{99 Cove ave., Milford, CT, 06461, USA \\
 E-mail: lsakhnovich@gmail.com}\\
 
 
\begin{abstract}
Basic relations for the mean length and algorithmic entropy are obtained by solving a new extremal problem. Using this extremal problem, they are 
obtained in a most simple and general way. The length and entropy are considered as two players of a new type of a game, in which we follow
the scheme of our previous work on thermodynamic characteristics in quantum and classical approaches. 
\end{abstract}

\section{Introduction} Algorithmic information theory (AIT) is an important and actively studied domain of  computer science (see, e.g., interesting results and
numerous references in \cite{1, Ch, LV, SeP, 9}).
AIT can be interpreted in terms  of statistical physics (SP) (see \cite{1,9, VP} and references therein). Let us introduce the corresponding notions from AIT and SP.

1. The set of all AIT programs  corresponds to the set of  energy eigenvectors
from SP.

2. The length $\ell_{k}$ of an AIT program  corresponds to the energy eigenvalue $E_{k}$ from SP. (Here and further $k \geq 1$.)

We denote by $P_{k}$ the probability that the length of the program is equal to
$\ell_{k}$, i.e. $P_{k}=P(\ell=\ell_{k}).$ Next, we introduce the notions
of the mean length $L$
(of programs)  and of the entropy $S$:
\begin{equation}\label{1.1}
L=\sum_{k}P_{k}\ell_{k},\qquad S=-\sum_{k}P_{k}\log {P_{k}}.
\end{equation}
 The connection between L and S we interpret in terms of game theory.
The necessity of the game theory approach can be explained in the following
way.  The notion of Gibbs ensemble is introduced in AIT  using an analogy
with the second law of thermodynamics:

\emph{ Gibbs ensemble maximizes  entropy
on the set of programs, where the  values $\{l_k\}$ and $L$ are fixed.}

So, the problem of a conditional
extremum appears. But the corresponding equation for the Lagrange multiplier is transcendental   and very complicated. Therefore, another argumentation is needed to find the
basic Gibbs formulas. This problem exists also  for  the SP case (see \cite[Ch.1, section 1]{2} and \cite[Ch.3, section 28]{3}).
 In the present note we use our approach the extremal SP problem 
 \cite{5,6,7,8}  to treat also to the corresponding AIT problem.
 Namely, we fix the Lagrange multiplier
$\beta=1/kT$. That is, we fix the AIT analogue $T$ of the temperature
from SP and introduce the compromise function
$F={-\beta}L+S$. Then the mean length $L$ and the entropy $S$ are two players of a game and the compromise result
is the extremum  point of the F.  Finally, we note, that the AIT analogue  of temperature was discussed
by K.Tadaki \cite{9}. He proved the following assertion:

If  the temperature is a computable positive number bounded by $1$,  it can be interpreted
as the {\it compression rate} in the AIT analogue of thermodynamic theory.

\section {Connection between length and entropy, \\
a game theoretical point of view}
Let the lengths   $\ell_{k}$ of the programs  be fixed. Consider the
mean length
$L$ and  the entropy $S$, which are given in \eqref{1.1}.
Note that $\sum_{k}P_{k}=1$.
Hence, $P_{k}$ can be represented in the  form
$P_{k}=p_{k}/Z$, where $Z=\sum_{k}p_{k}$.
Our aim is to find the probabilities $P_{k}$.
For that purpose we consider the function
\begin{equation}\label{2.1}
  F=\lambda{L}+S,
\end{equation}
where $\lambda=-\beta=-1/kT$.

{\bf Fundamental Principle.} {\it The function $F$ defines a game between
the mean length $L$ and the entropy $S$.}

To find the stationary point of $F$ we calculate
\begin{equation}\label{2.2}
  \frac{\partial{F}}{\partial{p_{j}}}=\lambda\Big(\ell_{j}/Z-\sum_{k=1}^{\infty}\ell_{k}p_{k}/Z^{2}\Big)-
(\log{p_{j})/Z}+\sum_{k=1}^{\infty}p_{k}\log{p_{k}}/Z^{2}.
\end{equation}
It follows from \eqref{2.2} that the point
\begin{equation}\label{2.3}
  p_{k}=e^{{\lambda}\ell_{k}}, \qquad k=1,2,\ldots
\end{equation}
is a stationary point.  Moreover, the stationary point is unique up to a scalar
multiple. Without loss of generality this multiple can be fixed as in \eqref{2.3}.
By direct calculation we get in the stationary point \eqref{2.3} the equalities
\begin{equation}\label{2.4}
  \frac{\partial^{2}F}{\partial{p_{k}^{2}}}=-Z_k/(p_kZ^2)<0,
\quad Z_k:= \sum_{j\not= k}p_j;
\quad \frac{\partial^{2}F}{\partial{p_{k}}\partial{p_{j}}}=1/Z^{2}>0,\quad j{\ne}k .
\end{equation}
Relations \eqref{2.4} imply the following assertion.

\begin{cor}
 The stationary point  \eqref{2.3} is a maximum  point of the function $F$.
 \end{cor}
\begin{proof} We shall use the following result (see 
\cite[Ch.7, Problem 7]{4}):
\begin{equation}\label{2.5}
\det\left[\begin{array}{ccccc}
                      r_{1} & a & a & ... & a \\
                      b & r_{2} & a & ... & a \\
                      b & b & r_{3} & ... & a \\
                      ... & ... & ... & ... & ... \\
                      b & b & b & ... & r_{n}
                    \end{array}\right]=\frac{af(b)-bf(a)}{a-b},\end{equation}
 where
 \begin{equation}\label{2.6}
 f(x)=( r_{1}-x)( r_{2}-x)...( r_{n}-x).\end{equation}
In the case that $a=b$ we have
\begin{equation}\label{2.7}
\det\left[\begin{array}{ccccc}
                      r_{1} & a & a & ... & a \\
                      a & r_{2} & a & ... & a \\
                      a & a & r_{3} & ... & a \\
                      ... & ... & ... & ... & ... \\
                      a & b & a & ... & r_{n}
                    \end{array}\right]=-af^{\prime}(a)+f(a).\end{equation}
Using \eqref{2.4} and \eqref{2.7} we can calculate the Hessian $H_{n}(F)$ in the stationary point:
\begin{equation}  \label{2.8}
H_{n}(F)=Z^{-2n}[-f^{\prime}(1)+f(1)],
\end{equation}
where $f$ is given by \eqref{2.6} and $r_{k}=-Z_{k}/p_{k}<0$. Rewrite \eqref{2.8}
in the form
\[
H_{n}(F)=(-Z)^{-n}\Big(1-\big(\sum_{k=1}^n p_k\big)/Z\Big)/\prod_{k=1}^n p_k
\]
to see that the relation
${\mathrm{sgn}\,}\big(H_{n}(F)\big)=(-1)^{n}$
is true. Hence, the corollary is proved.
\end{proof}
So, we proved the proposition below.
\begin{Pn} The mean length and entropy satisfy relations
\begin{align}\label{2.9}&
L=\sum_{k}\ell_{k}e^{\lambda\ell_{k}}/Z, \\
\label{2.10}&
S=-\sum_{k}(e^{\lambda\ell_{k}}/Z)\log (e^{\lambda\ell_{k}}/Z), 
\end{align}
where $Z= \sum_{k}e^{\lambda\ell_{k}}$. 
\end{Pn}
Note that  the basic relations \eqref{2.3}, \eqref{2.9},  and \eqref{2.10}
are obtained by solving a new extremal problem. Namely, in the introduced function F
the parameter $\lambda$ is fixed  instead of the  length $L$, which is usually fixed.


\end{document}